\documentclass{LMCS}

\def\doi{9(1:16)2013}
\lmcsheading%
{\doi}
{1--16}
{}
{}
{Mar.~16, 2012}
{Mar.~31, 2013}
{}

\usepackage{amssymb,latexsym}
\usepackage{graphicx, amsmath, amssymb, pstricks}
\usepackage{bussproofs}
\usepackage{hyperref}
\usepackage{phuong}
\usepackage{pdfsync}
\usepackage{pifont}

\usepackage{mathrsfs} 
\usepackage{enumerate}
\usepackage{stmaryrd}

\newcommand{\Nat}{\ensuremath \mathbb N}

\newcommand{\mbf}[1]{\mbox{{#1}}}   

\newcommand{\abs}[1]{\left\vert#1\right\vert}
\newcommand{\set}[1]{\left\{#1\right\}}

\newcommand {\cd}{\cdot}

\DeclareMathAlphabet{\mathitbf}{OML}{cmm}{b}{it}

\newcommand{\NP}{\mbf{NP}}

\font\sf=cmss10
\newcommand{\Nats}{{\hbox{\sf I\kern-.13em\hbox{N}}}}   
\newcommand{\Reals}{{\hbox{\sf I\kern-.14em\hbox{R}}}}  
\newcommand{\Ints}{{\hbox{\sf Z\kern-.43emZ}}}          
\newcommand{\CC}{{\hbox{\sf C\kern -.48emC}}}           
\newcommand{\QQ}{{\hbox{\sf C\kern -.48emQ}}}           

\renewcommand{\And}{\land}

\newcommand{\dotminus}{\mathbin{\setbox0\hbox{$-$}\setbox2\hbox to\wd0{\hss$^{\mkern1mu\cdot}$\hss}\wd2=0pt\box2\box0}}

\ifx\theorem\undefined

\newtheorem{theorem}{Theorem}[section]

\newtheorem{proposition}[theorem]{Proposition}
\newtheorem{corollary}[theorem]{Corollary}
\newtheorem{definition}{Definition}[section]

\fi

\ifx\proof\undefined
\newenvironment{proof}{\QuadSpace\par\noindent{\bf
Proof}:}{\EndProof\HalfSpace} \fi

\newcommand{\QuadSpace}{}
\newcommand{\HalfSpace}{}
\newcommand{\EndProof}{ \hfill \vrule width 1ex height 1ex depth 0pt }

\def\RL0{{\mbox{\rm R(lin)}}}
\def\RZ0{{\mbox{\rm R$^0$(lin)}}}
\def\RC0{R(lin) with constant coefficients}
\def\RCD0#1{{\mbox{\rm R$_{#1}$(lin)}}}

\def\Tse0{{\mbox{$\neg$\textsc{Tseitin}$_{G,p}$}}}

\definecolor{bluetxt}{rgb}{0,0,.5}
\definecolor{myred}{rgb}{0.6,0.0,0.1}
\definecolor{greentxt}{rgb}{0,.5,0}
\definecolor{redtxt}{rgb}{0.1,0.1,0.65}
\definecolor{purpletxt}{rgb}{0.6,0.1,0.7}
\definecolor{black}{rgb}{.0,.0,.0}
\definecolor{verydarkblue}{rgb}{.0,.0,.2}
\definecolor{lightgray}{rgb}{.7,.7,.7}
\definecolor{bgcolor}{rgb}{.8,.8,.5}
\definecolor{lightkhaki}{rgb}{0.945,.946,.355}

\ifx\proof\undefined
\newenvironment{proof}{

\smallskip
\noindent\emph{Proof.}}{\hfill\(\Box\)
\bigskip
} \fi

\newlength{\defbaselineskip}
\newcommand{\Comment}[1]{}
\newcommand {\mar}[1]{}
\renewcommand{\footnote}[1]{}

\title{Polylogarithmic Cuts in Models of \VZ}

\author[S.~M\"uller]{Sebastian M\"uller}
\address{Faculty of Mathematics and Physics \\
        Charles University,  Prague}
\email{muller@karlin.mff.cuni.cz}
\thanks{Supported by the Marie Curie Initial
    Training Network in Mathematical Logic -  MALOA - From MAthematical LOgic to Applications, PITN-GA-2009-238381}


\ACMCCS{[{\bf Theory of computation}]: Logic---Proof theory}

\begin{document}

\begin{abstract}
  We study initial cuts of models of weak two-sorted Bounded
  Arithmetics with respect to the strength of their theories and show
  that these theories are stronger than the original one. More
  explicitly we will see that polylogarithmic cuts of models of $\VZ$
  are models of $\VNC^1$ by formalizing a proof of Nepomnjascij's
  Theorem in such cuts.  This is a strengthening of a result by Paris
  and Wilkie.

  We can then exploit our result in Proof Complexity to observe that
  Frege proof systems can be sub exponentially simulated by bounded
  depth Frege proof systems. This result has recently been obtained by
  Filmus, Pitassi and Santhanam in a direct proof. As an
  interesting observation we also obtain an average case separation of
  Resolution from $\mathsf{AC}^0$-Frege by applying a recent result
  with Tzameret.
\end{abstract}

\maketitle

\section{Introduction}

This article is on the one hand on models of weak arithmetics and on the other on proof complexity, i.e. the
question of how long formal proofs of tautologies have to be in given proof systems. Therefore the introduction
will consist of two parts, one for each subject.

Models of weak arithmetics, like $I\Delta_0$, have been extensively studied for several reasons. They are
possibly the simplest objects whose theories bear enough strength to do a good part of mathematics in, yet they
are weak enough to allow for a certain kind of constructiveness. The latter has been demonstrated over and over
again by various results connecting weak arithmetic theories with complexity classes and computability. We are
interested in the strength of the theory obtained by restricting our objects of reasoning to a small initial
part of a given model. Since a two-sorted theory, such as \VZ, is much stronger on its number part than on its
set part, it is likely that such a cut is a model of a supposedly much stronger theory. Indeed we will see in
Section~\ref{Sec:Polylog Cuts} that certain cuts of models of \VZ\ are models of the provably stronger theory
$\VNC^1$. This strengthens a result by Paris and Wilkie \cite{PW87}\cite{PW85}, who show that such cuts are
models of $\mathbf{VTC}^0$. In fact they work in a more general setting and, following our argumentation, their
result readily implies the sub exponential simulation of $\mathsf{TC}^0$-Frege by $\mathsf{AC}^0$-Frege from
Bonet, Domingo, Gavald\`a, Maciel, and Pitassi \cite{BDGMP04}.

Proof Complexity, on the other hand, more or less began when Cook and Reckhow \cite{CR79} discovered the close
connection between the lengths of formal proofs in propositional proof systems and standard complexity classes.
This connection yields a possibility of dealing with the $\mathsf{coNP/NP}$ question by asking, whether there
exists a propositional proof system that is polynomially bounded. We will not directly address this question
here, but rather explore the relative strengths of two major proof systems, Frege and bounded depth Frege. These
proof systems have been extensively studied, due to their natural appearance as classical calculi, such as
Gentzen's PK, and it is well known that Frege systems are stronger than bounded depth Frege systems, as the
former system has polynomial size proofs for the Pigeonhole Principle (see \cite{Bu87}), while the latter does
not (see \cite{KPW95} and \cite{PBI93}). Lately, Filmus, Pitassi and Santhanam \cite{FPS11} have proved a sub
exponential simulation of Frege by bounded depth Frege using a combinatoric argument. In Section~\ref{Sec:Impl
Proof Cmpl} we will obtain the same result by an application of our result about cuts to the provability of the
Reflection Principle for Frege in bounded depth Frege. Currently Cook, Ghasemloo and Nguyen \cite{CGN12} are
working on a purely syntactical proof that gives a slightly better result with respect to the strength of the
simulated proof system.

The paper is built-up as follows. In section \ref{Sec:Preliminaries} we briefly recapture some basics about
Complexity Theory, Bounded Arithmetic, Proof Complexity and the various connections between them. As this is
only expository it might be helpful to consult some of the references for a more detailed introduction (see
\cite{AB09}, \cite{CN10} and \cite{Kra95}). After that, in Section~\ref{Sec:Polylog Cuts} we prove a
formalization of Nepomnjascij's Theorem in the polylogarithmic cut of a model of $\VZ$. Using a standard
algorithm for evaluating circuits and then applying the formalized version of Nepomnjascij's Theorem we can
conclude that this cut is indeed a model of $\VNC^1$. Finally, in Section~\ref{Sec:Impl Proof Cmpl}, we apply
this result to prove that a version of the Bounded Reflection Principle of Frege is provable in $\VZ$. This,
together with a standard argument linking the provability of Reflection Principles with simulation results,
yields the sub exponential simulation of Frege by bounded depth Frege.

\section{Preliminaries}\label{Sec:Preliminaries}

We assume familiarity with Turing machines, circuits and standard complexity classes such as $\mathsf{P}$,
$\mathsf{NP}$, $\mathsf{TimeSpace}(f,g)$, $\mathsf{NC}^i$, $\mathsf{AC}^i$ and so on. See for example
\cite{AB09} for an introduction. We will not work a lot within these classes, but rather apply known relations
between such classes and weak arithmetic theories.

We will work in a two-sorted arithmetic setting, having one sort of variables representing numbers and the
second sort representing bounded sets of numbers. We identify such bounded subsets with strings. See \cite{CN10}
for a thorough introduction. The underlying language, denoted \LTwoA, consists of the following relation,
function and constant symbols:\[ \set{+,\cd,\le, 0,1,\abs{\cdot},=_1,=_2,\in} \,.\]

An  \LTwoA-structure $M$ consists of a first-sort universe $U_1^M$ of numbers and a second-sort universe $U_2^M$
of bounded subsets of numbers.  If $M$ is a model of the two-sorted theory $\VZ$ (see \ref{sec VZ}), then the
functions $+$ and $\cd$ are the addition and multiplication on the universe of numbers. $0$ and $1$ are
interpreted as the appropriate elements zero and one with respect to addition and multiplication. The relation
$\le$ is an ordering relation on the first-sort universe. The function $\abs{\cdot }$ maps an element of the set
sort to its largest element plus one (i.e. to an element of the number sort). The relation $=_1$ is interpreted
as equality between numbers, $=_2$ is interpreted as equality between bounded sets of numbers. The relation
$\in$ holds for a number $n$ and a set of numbers $N$ if and only if $n$ is an element of $N$. The standard
model of two-sorted Peano Arithmetic will be denoted as $\Nat_2$. It consists of a first-sort universe
$U_1=\Nat$ and a second-sort universe $U_2$ of all finite subsets of $\Nat$. The symbols are interpreted in the
usual way.

We denote the first-sort (number) variables by lower-case letters $ x,y,z,\dots,\alpha,\beta,\gamma,\dots$, and the second-sort (set)
variables by capital letters $ X,Y,Z,\dots,A,B,\Gamma,\dots$. In case it helps to describe the meaning of a variable we will use lower case words for first-sort and words starting with a capital letter for second-sort variables. We can build formulas in the usual way, using two sorts of
quantifiers, number quantifiers and string quantifiers. A number quantifier $\exists x$ ($\forall x$) is bounded
if it is of the form $\exists x (x\leq f\land\dots)$ ($\forall x (x\leq f\rightarrow\dots)$) for some number
term $f$. A string quantifier $\exists X$ ($\forall X$) is bounded if it is of the form $\exists X (\abs{X}\leq
f\land\dots)$ ($\forall X (\abs{X}\leq f\rightarrow\dots)$) for some number term $f$. A formula is bounded iff
all its quantifiers are. All formulas in this paper will be bounded. A formula $\varphi $ is in $\Sigma^B_0$ (or
$\Pi^B_{0}$) if it uses no string quantifiers and all number quantifiers are bounded. A formula $\varphi $ is a
$\Sigma^B_{i+1}$ (or $\Pi^B_{i+1}$) if it is of the form $\exists X_1\leq p(n)\dots\exists X_m\leq p(n) \psi$
(or $\forall X_1\leq p(n)\dots\forall X_m\leq p(n) \psi$), where $\psi\in\Pi^B_i$ (or $\psi\in\Sigma^B_i$,
respectively). If a relation or predicate can be defined by both a $\Sigma^B_i$ and a $\Pi^B_i$ formula, then we
call it $\Delta^B_i$ definable. The {\em depth} of a formula is the maximal number of alternations of its
logical connectives and quantifiers.

As mentioned before we will represent a bounded set of numbers $N$ by a finite string $S_N=S^0_N\dots
S^{\abs{N}-1}_N$ such that $S^i_N=1$ if and only if $i\in N$. We will abuse notation and identify bounded sets
and strings, i.e. $N$ and $S_N$.

Further, we will encode monotone propositional formulas inductively as binary trees in the standard way, giving
a node the value $1$ if it corresponds to a conjunction and the value $0$, if it corresponds to a disjunction.
Binary trees are encoded as strings as follows. If position $x$ contains the value of a node $\mathfrak{n}_x$,
then the value of its left successor is contained in position $2x$, while the value of its right successor is in
$2x+1$.

\subsection{Elements of Proof Complexity}\label{Sec:Elements Proof Cmpl}

We restate some basic definitions introduced in \cite{CR79}.

\begin{definition}\label{def pps}
  A {\em propositional proof system (pps)} is a surjective polynomial-time function
  $P:\{0,1\}^*\longrightarrow\mathsf{TAUT}$, where $\mathsf{TAUT}$ is the set of propositional tautologies (in
  some natural encoding). A string $\pi$ with $P(\pi)=\tau$ is called a {\em $P$-proof} of $\tau$.
\end{definition}

We can define a quasi ordering on the class of all pps as follows.

\begin{definition}\label{def:simulation}
  Let $P,Q$ be propositional proof systems.
  \begin{iteMize}{$\bullet$}
    \item $P$ simulates $Q$ (in symbols $P\geq Q$), iff there is a polynomial $p$, such that for all $\tau\in\mathsf{TAUT}$ there is a
    $\pi_P$ with $P(\pi_P)=\tau$, such that for all $\pi_Q$ with $Q(\pi_Q)=\tau$, $\abs{\pi_P}\leq
    p(\abs{\pi_Q})$.
    \item If there is a polynomial time machine that takes $Q$-proofs and produces $P$-proofs for the same
    formula we say that
    $P$ p-simulates $Q$ (in symbols $P\geq_p Q$).
    \item If $P$ and $Q$ mutually (p-)simulate each other, we say that they are (p-)equivalent (in symbols
    $P\equiv Q$ and $P\equiv_p Q$, respectively).
  \end{iteMize}\smallskip
\end{definition}

\noindent In this article we will be mainly interested in bounded depth Frege systems and some of their extensions. A
Frege system is a typical textbook proof system, such as Gentzen's propositional calculus PK. We will only
sketch a single rule of such a system as an example and refer the interested reader to standard logic textbooks.

   \begin{prooftree}
     \AxiomC{$ \Gamma\longrightarrow A,\Delta$}
     \AxiomC{$ \Gamma,A\longrightarrow \Delta$}
     \RightLabel{(Cut)}
     \BinaryInfC{$ \Gamma\longrightarrow\Delta$}
   \end{prooftree}\smallskip
Here, $\Delta$ and $\Gamma$ are sets of formulas while $A$ is a formula. $\Gamma\longrightarrow\Delta$ is read
as "The conjunction of all formulas in $\Gamma$ implies the disjunction of all formulas in $\Delta$". The Cut
Rule therefore says that, if $\Gamma$ implies $A$ or $\Delta$, and $\Gamma$ and $A$ imply $\Delta$, then
$\Gamma$ already implies $\Delta$. The formula $A$ is called the {\em Cut Formula}.

In a bounded depth Frege system the depths of all formulas in a derivation are bounded by some global constant.
This is equivalent to being representable by an $\mathsf{AC}^0$ circuit. Thus we also call bounded depth Frege
$\mathsf{AC}^0$-Frege. If the formulas are unbounded, we speak of $\mathsf{NC}^1$-Frege or simply of Frege. We
readily get

\begin{fact}
  $\mathsf{AC}^0 \mbox{-Frege} \leq_p \mbox{Frege}$.
\end{fact}

A pps $P$ is {\em polynomially bounded} iff there is a polynomial $p$ such that every tautology $\tau$ has a
$P$-proof $\pi$ with $\abs{\pi}\leq p(\abs{\tau})$.

We are interested in the existence of polynomially bounded pps. This is, at least in part, due to the following
theorem.

\begin{fact}[\cite{CR79}]
  $\mathsf{NP=coNP} \Leftrightarrow \mbox{There exists a polynomially bounded pps.}$
\end{fact}

An easier task than searching for a polynomially bounded pps might be to find some pps with sub exponential
bounds to the lengths of proofs. This corresponds to the question, whether sub exponential time nondeterministic
Turing machines can compute $\mathsf{coNP}$-complete languages. To explore the existence of such systems we
generalize Definition~\ref{def:simulation}.

\begin{definition}
    Let $P,Q$ be propositional proof systems and $F$ a family of increasing functions on $\Nat$.
  \begin{iteMize}{$\bullet$}
    \item $P$ $F$-simulates $Q$ (in symbols $P\geq^F Q$), iff there is a function $f\in F$, such that for all $\tau\in\mathsf{TAUT}$ there is a
    $\pi_P$ with $P(\pi_P)=\tau$, such that for all $\pi_Q$ with $Q(\pi_Q)=\tau$, $\abs{\pi_P}\leq
    f(\abs{\pi_Q})$.
    \item If there is an $\mathsf{Time}(F)$-machine that that takes $Q$-proofs and produces $P$-proofs for the same
    formula we say that
    $P$ $F$-computably simulates $Q$ (in symbols $P\geq^F_p Q$).
    \item If $P$ and $Q$ mutually $F$-(computably) simulate each other, we say that they are $F$-(computably) equivalent (in symbols
    $P\equiv^F Q$ and $P\equiv^F_p Q$, respectively).
  \end{iteMize}\smallskip
\end{definition}

\noindent We say a pps $P$ {\em sub exponentially simulates} a pps $Q$ iff the above $F$ can be chosen as a class of
$2^{n^{o(1)}}$ functions.

\subsection{The theory \VZ}\label{sec VZ}
The base theory we will be working with is $\VZ$. It consists of the following axioms:
\begin{center}
\framebox[\textwidth]{
\parbox{335pt}{
\begin{align*}
& \textbf{Basic 1}.\  x+1\neq 0     & \textbf{Basic 2}.\ x+1=y+1\rightarrow
x=y \\
& \textbf{Basic 3}.\   x+0=x        & \textbf{Basic 4}.\ x+(y+1)=(x+y)+1
\\
& \textbf{Basic 5}.\   x\cdot 0=0   & \textbf{Basic 6}.\ x\cdot(y+1)=(x\cdot
y)+x \\
& \textbf{Basic 7}.\   (x\leq y \land y\leq x)\rightarrow x=y
                                & \textbf{Basic 8}.\   x\leq x+y \\
& \textbf{Basic 9}.\   0\leq x      & \textbf{Basic 10}.\   x\leq y\vee
y\leq x    \\
& \textbf{Basic 11}.\   x\leq y\leftrightarrow x<y+1
                                & \textbf{Basic 12}.\   x\neq
                                0\rightarrow\exists
                                y\leq  x(y+1=x)
                                                                    \\
& \textbf{L1}.\  X(y)\rightarrow y<\abs{X}
                                & \textbf{L2}.\  y+1=\abs{X}\rightarrow
                                X(y)
\end{align*}
\begin{equation*}\begin{split}
\mbox{\bf{SE}. }(\abs{X}=\abs{Y}\land \forall i\leq\abs{X}
    (X(i)\leftrightarrow Y(i)))\rightarrow X=Y\\
\mbox{{\bf $\Sigma^B_0$-COMP.\ }}   \exists X\leq y\forall z<y (X(z)\leftrightarrow \varphi (z))\,,\quad
\mbox{for all}\ \varphi \in\Sigma^B_0\,.
\end{split}
\end{equation*}
   }
  }
\end{center}
Here, the Axioms {\bf Basic 1} through {\bf Basic 12} are the usual axioms used to define Peano Arithmetic
without induction ($\mathsf{PA^-}$), which settle the basic properties of Addition, Multiplication, Ordering,
and of the constants 0 and 1. The Axiom {\bf L1} says that the length of a string coding a finite set is an
upper bound to the size of its elements. {\bf L2} says that $\abs{X}$ gives the largest element of $X$ plus $ 1
$. {\bf SE} is the extensionality axiom for strings which states that two strings are equal if they code the
same sets. Finally, {\bf $\Sigma^B_0$-COMP} is the comprehension axiom schema for $\Sigma^B_0$-formulas (it is
an axiom for each such formula) and implies the existence of all sets, which contain exactly the elements that
fulfill any given $\Sigma^B_0$ property.

\begin{fact}
The theory \VZ\ proves the Induction Axiom schema for $\Sigma^B_0$ formulas $ \Phi $:
\[
    (\Phi(0)\And \forall x(\Phi(x)\rightarrow \Phi(x+1)) ) \rightarrow \forall z\,\Phi(z).
\]
\end{fact}

When speaking about theories we will always assume that the theories are two-sorted theories as in \cite{CN10}.

%
The following is a basic notion:
\begin{definition}[Two-sorted definability]\label{def:Two-sorted definabilty}
Let $\mathcal T $ be a theory over the language $ \mathcal L \supseteq \LTwoA$ and let $ \Phi $ be a set of
formulas in the language $ \mathcal L $. A number function $f$ is $\Phi $-definable in a theory $\mathcal T$ iff
there is a formula $ \varphi(\vec x,  y,\vec X)$ in $ \Phi $ such that $\mathcal T$ proves
  \begin{equation*}
    \forall\vec x\forall\vec X\exists!y\varphi(\vec x, y,\vec X)
  \end{equation*}
  and it holds that
  \begin{equation}
    \label{eq:def axiom num}
    y=f(\vec x,\vec X)\leftrightarrow \varphi(\vec x,y,\vec X).
  \end{equation}
A string function $F$ is $\Phi $-definable in a theory $\mathcal T$ iff there is a formula $\varphi(\vec
  x,\vec X,Y)$ in $ \Phi $ such that $\mathcal T$ proves
  \begin{equation*}
    \forall\vec x\forall\vec X\exists!Y\varphi(\vec x,\vec X,Y)
  \end{equation*}
  and it holds that
  \begin{equation}
    \label{eq:def axiom str}
    Y=F(\vec x,\vec X)\leftrightarrow \varphi(\vec x,\vec X,Y).
  \end{equation}
Finally, a relation $R(\vec x,\vec X)$ is $\Phi $-definable iff there is a formula $\varphi(\vec x,\vec X)$ in $
\Phi $ such that it holds that
  \begin{equation}
    \label{eq:def axiom rel}
    R(\vec x,\vec X)\leftrightarrow \varphi(\vec x,\vec X).
  \end{equation}
\end{definition}

Moreover we wish to talk about sequences coded by strings or numbers. For a string $X$ we let $X[i]$ be the
$i$th bit of $X$. Assuming a tupling function $\langle \cdot,\dots,\cdot\rangle$ we can also talk of $k$-ary
relations for any constant $k$. We refer to $X[\langle i_0,\dots ,i_{k-1} \rangle]$, to say that the objects
$i_0,\dots,i_{k-1}$ are in the relation $X$ (which is equivalent to saying that the predicate $X$ holds for the
number $\langle i_0,\dots,i_{k-1}\rangle$, i.e. that the $X$ contains that number as an element). For the sake
of simplicity we also refer to $X[\langle i_0,\dots ,i_k \rangle]$ by $X[i_0,\dots ,i_k]$.

Using $k$-ary relations we can also encode sequences of bounded numbers $x_0,\dots,x_m$ by $x_i=X[\langle
i,0\rangle]X[\langle i,1\rangle]\dots X[\langle i,k\rangle]$ in binary. Matrices and so on can obviously be
formalized in the same way.

Given a string $X[\langle x_1,\dots, x_k\rangle]$ representing a $k$-ary relation, we denote the $k-\ell$-ary
substring with parameters $a_{i_1},\dots, a_{i_{\ell}}$ by $X[\langle
\cdot,\dots,\cdot,a_{i_1},\cdot,\dots,a_{i_{\ell}},\cdot,\dots,\cdot\rangle]$.
For example we refer to the element $a_{ij}$ of a given matrix $A[\langle
x_1,x_2,x_3\rangle]$ as $A[\langle i,j,\cdot\rangle]$, a string representing $a_{ij}$ in binary. Observe that
this substring can be $\Sigma^B_0$ defined in $\VZ$.

Given a number $x$ we denote by $\langle x\rangle_j$ the $j$th number in the  sequence encoded by $x$. To do
this we assume a fixed $\Sigma^B_0$ definable encoding of numbers that is injective. The sequence itself will be
addressed as $\langle x\rangle$. As above, we can also talk about matrices, etc. in this way, i.e. read such a
sequence as a sequence of $k$-tuples.

We want to identify strings of short length with sequences of numbers. Thus, given a string $X$ of length $O(n)$
we can $\Sigma^B_0$-define (in $\VZ$) a number $x\leq 2^{O(n)}$ that codes a sequence $\langle x\rangle$, such
that $X[i]=\langle x\rangle_i$ for all $i<\abs{X}$ and vice versa. We will use $\langle x\rangle\approx X$ and
$\langle x\rangle_i\approx X[i]$ to denote the above identification. Observe that $n$ has to be very small in
order to be able to do the above in $\VZ$.

\subsubsection{Computations in models of $\VZ$}\label{Secsub:Comp in VZ}

Given a polynomially bounded Turing machine $A$ in a binary encoding, we can $\Sigma^B_1$ define a predicate
$\mathsf{ACC}_A(X)$, that states that $X$ is accepted by $A$. This can readily be observed, since, provided some
machine $A$, there is a constant number of states $\sigma_1,\dots, \sigma_k$ and the whole computation can be
written into a matrix $W$ of polynomial size. That $W$ is indeed a correct computation can then be easily
checked, because the computations are only local.

More precisely let $A=\langle\sigma_1,\dots,\sigma_k; \delta\rangle$ be given, where the $\sigma_i$ are
different states, with $\sigma_1$ being the initial state and $\sigma_k$ being the accepting state and $\delta$
is the transition function with domain $\{\sigma_1,\dots,\sigma_k\}\times\{0,1\}$ and range
$\{\sigma_1,\dots,\sigma_k\}\times\{0,1\}\times\{\leftarrow,\downarrow,\rightarrow\}$, which describes what the
machine does. I.e. if $\delta(a,b)=(c,d,e)$, then if the machine is in state $a$ and reads $b$, it replaces $b$
by $d$, goes into state $c$ and moves one position on the tape in the direction $e$. For our formalization we
will assume a function $\delta:\Nat^2\rightarrow\Nat$ and interpret it in the following way,
$\delta(\sigma_a,b)=(\langle\delta(\sigma_a,b)\rangle_1,\langle\delta(\sigma_a,b)\rangle_2,
\langle\delta(\sigma_a,b)\rangle_3)$, where we identify $\downarrow=0,\leftarrow=1,\rightarrow=2$.

Let the polynomial $p$ bound the running time of $A$, then we can formalize $\mathsf{ACC}_A(X)$ as follows
\begin{equation}\label{Def ACC}
  \begin{split}
    \exists W\leq (p&(\abs{X})^2\cdot\log (k)^2)\forall i,i'\leq p(\abs{X})
    \forall 0<\alpha\leq k(\\ &
    i<\abs{X}\rightarrow(\langle W[\langle 0,i,\cdot \rangle]\rangle_1 =X[i]\land i> 0\rightarrow \langle
    W[\langle 0,i,\cdot \rangle]\rangle_2 = 0 \land \langle W[\langle 0,0,\cdot
    \rangle]\rangle_2=1)\land\\ & i\geq\abs{X}\rightarrow(\langle W[\langle 0,i,\cdot \rangle]\rangle_1 =0 \land
    \langle W[\langle 0,i,\cdot\rangle]\rangle_2=0)\land\\ & \langle W[\langle j,i,\cdot \rangle]\rangle_2
    =0\rightarrow (\langle W[\langle j+1,i,\cdot \rangle]\rangle_1=
    \langle W[\langle j,i,\cdot \rangle]\rangle_1)\land\\ &
    \langle W[\langle j,i,\cdot \rangle]\rangle_2
    =\alpha\rightarrow (\langle W[\langle j+1,i,\cdot \rangle]\rangle_1=
    \langle\delta(\alpha,\langle W[\langle j,i,\cdot \rangle]\rangle_1)\rangle_2\land\\ &
    \hspace{0.2cm}(\langle\delta(\alpha,\langle W[\langle j,i,\cdot \rangle]\rangle_1)\rangle_3=0 \rightarrow
    \langle W[\langle j+1,i,\cdot \rangle]\rangle_2=
    \langle\delta(\alpha,\langle W[\langle j,i,\cdot \rangle]\rangle_1)\rangle_1)\land\\ &
    \hspace{0.2cm}(\langle\delta(\alpha,\langle W[\langle j,i,\cdot \rangle]\rangle_1)\rangle_3=1 \rightarrow
    \langle W[\langle j+1,i\dotminus 1,\cdot \rangle]\rangle_2=
    \langle\delta(\alpha,\langle W[\langle j,i,\cdot \rangle]\rangle_1)\rangle_1)\land\\ &
    \hspace{0.2cm}(\langle\delta(\alpha,\langle W[\langle j,i,\cdot \rangle]\rangle_1)\rangle_3=2 \rightarrow
    \langle W[\langle j+1,i+1,\cdot \rangle]\rangle_2=
    \langle\delta(\alpha,\langle W[\langle j,i,\cdot \rangle]\rangle_1)\rangle_1))\land\\ &
    (i\neq i'\rightarrow (\langle W[\langle j,i,\cdot \rangle]\rangle_2>0\rightarrow
    \langle W[\langle j,i',\cdot \rangle]\rangle_2=0))\land
    W[\langle p(\abs X),i,\cdot \rangle]\rangle_2 = k).
  \end{split}
\end{equation}

\noindent Thus, in plain English, $\mathsf{ACC}_A(X)$ says that there exists a matrix $W$ of pairs of numbers that
witnesses an accepting computation of $A$. Here, $\langle W[\langle i,j,\cdot \rangle]\rangle$ is supposed to
code the $j$th cell on the Turing machine's tape after $i$ steps of computations on input $X$. As noted above,
$\langle W[\langle i,j,\cdot \rangle]\rangle_1$ is a binary number, which is the value of the cell, $\langle
W[\langle i,j,\cdot \rangle]\rangle_2$ is a number coding the state the machine is in iff the pointer is on that
cell.

The second and third line of the definition say that the tape in the initial step contains $X$ padded with
zeroes in the end to get the proper length $(p(\abs{X}))$ and that the read/write head is in its starting state
and position. The fourth line says that if the read/write head is not on cell $i$, then nothing happens to the
content of cell $i$. The fifth line says that the content of the cell, where the read/write head is in step $j$,
is changed according to $\delta$. The next three lines tell us where the read/write head moves. The last line
says that there is at most one position on the tape where the read/write head may be at any step and that the
state after the last step is accepting.

We also define a $\Sigma^B_1$-predicate $\mathsf{REACH}_A(Y,Y')$ that says that $A$ reaches configuration $Y'$
from configuration $Y$ in at most $p(\abs{Y})$ many steps. This is essentially the same predicate as
$\mathsf{ACC}$, with the constraints on the initial and accepting state lifted and instead a constraint added
that the first line of computation is $Y$ and the last is $Y'$. We omit the details as it does not severely
differ from the above definition of $\mathsf{ACC}$.


\subsection{Extensions of \VZ}\label{Sec:Ext of VZ}

The Theory \VZ\ serves as our base theory to describe complexity classes by arithmetical means.

The problem, whether a given monotone formula $\varphi$ of size $\ell$ and depth $\lceil\log(\ell)\rceil$ is
satisfiable under a given assignment $I$ is $\mathsf{AC}^0$-complete for $\mathsf{NC}^1$. Therefore Cook and
Nguyen (\cite{CN10}) define the class $\VNC^1$ as $\VZ$ augmented by the axiom $MFV\equiv \exists Y\leq
2a+1.\delta_{MFV}(a,G,I,Y)$, where
\begin{equation*}
 \begin{split}
    \delta_{MFV}(a,G,I,Y)\equiv \forall x<a(&(Y(x+a)\leftrightarrow I(x))\land Y(0)\land\\
    0<x\rightarrow(Y(x)\leftrightarrow(&(G(x)\land Y(2x)\land
    Y(2x+1))\vee\\& (\neg G(x)\land (Y(2x)\vee Y(2x+1)))))).
  \end{split}
\end{equation*}
So, $MFV$ states that there is an evaluation $Y$ of the monotone formula represented by $G$ under the assignment
given by $I$ of length at most $2a+1$. More specifically, $G$ is a tree-encoding of the formula, where $G(x)$ is
true, if node $x$ is $\land$ and false, if $x$ is $\vee$. The evaluation $Y$ takes the value of the variables
given by $I$ and then evaluates the formula in a bottom-up fashion using a standard tree encoding. Thus, the
value of the formula can be read at $Y(1)$.

It is interesting to observe that $MFV$ does not hold in $\VZ$. This is, since an application of the Witnessing
Theorem for $\VZ$ to a proof of $MFV$ would yield an $\mathsf{AC}^0$-definition of satisfaction for monotone
$\mathsf{NC}^1$ circuits. This implies that monotone $\mathsf{NC}^1\subseteq\mathsf{AC}^0$, which is known to be
false.

\subsection{Relation between Arithmetic Theories and Proof Systems}\label{Sec:Rel BA pps}

In this section we will remind the reader of a connection between the Theory $\VZ$ and some of its extensions
and certain propositional proof systems (see also \cite{CN10}\cite{Kra95}).

\begin{definition}The following predicates will be subsequently used. They are definable with respect to $\VZ$
 (see \cite{Kra95}).
  \begin{iteMize}{$\bullet$}
    \item $\mathsf{Fla}(X)$ is a $\Sigma^B_0$ formula that says that the the string $X$ codes a
    formula.
    \item $Z\models X$ is the $\Delta^B_1$ definable property that the truth assignment $Z$ satisfies the formula $X$.
    \item $\mathsf{Taut}(X)$ is the $\Pi^B_1$ formula $\mathsf{Fla}(X)\land \forall Z\leq t(\abs{X}) Z\models
    X$, where $t$ is an upper bound to the number of variables in formulas coded by strings of length $\abs{X}$.
    \item $\mathsf{Prf}_{F_d}(\Pi,A)$ is a $\Sigma^B_0$ definable predicate meaning $\Pi$ is a
    depth $d$ Frege proof for $A$.
    \item $\mathsf{Prf}_{F}(\Pi,A)$ is a $\Sigma^B_0$ definable predicate meaning $\Pi$ is a Frege
    proof for $A$.
  \end{iteMize}
\end{definition}

The following holds

\begin{fact}[see \cite{CN10}]\label{fact:VZ and bd Frege}
  The Theory $\VZ$ proves that $\mathsf{AC}^0$-Frege is sound, i.e. for every $d$ $$\forall A\forall\Pi
  \mathsf{Prf}_{F_d}(\Pi,A)\rightarrow \mathsf{Taut}(A).$$
\end{fact}

\begin{fact}[see \cite{CN10}]\label{fact:VNC1 and NC1 Frege}
  The Theory $\VNC^1$ proves that Frege is sound, i.e. $$\forall A\forall\Pi
  \mathsf{Prf}_{F}(\Pi,A)\rightarrow \mathsf{Taut}(A).$$
\end{fact}

On the other hand, provability of the universal closure of $\Sigma^B_0$ formulas in $\VZ$ and $\VNC^1$ implies
the existence of polynomial size proofs of their propositional translations in $\mathsf{AC}^0$-Frege and Frege,
respectively.

The propositional translation $\llbracket\varphi(\bar x,\bar X)\rrbracket_{\bar m,\bar n}$ of a $\Sigma^B_0$
formula $\varphi(\bar x,\bar X)$ is a family of propositional formulas built up inductively (on the logical
depth) as follows. If $\varphi$ is atomic and does not contain second sort variables, we evaluate $\varphi$ in
$\Nat_2$, if it contains second sort variables, we have to introduce propositional variables. If $\varphi$ is a
boolean combination of formulas $\psi_i$ of lower depth, the translation is simply the same boolean combination
of the translations of the $\psi_i$. If $\varphi$ is $\exists\psi$ or $\forall\psi$ we translate it to the
disjunction or conjunction of the translations, respectively. For a proper definition see \cite{CN10}.

\begin{fact}
  \label{fact:polynomial simulation}
  There exists a polynomial $p$ such that for all $\Sigma^B_0$ formulas $\varphi(\bar x,\bar X)$ the
  following holds
  \begin{iteMize}{$\bullet$}
    \item If $\VZ\vdash\forall\bar X\forall \bar x\varphi(\bar x,\bar X)$, then there exist bounded depth Frege
    proofs of all $\llbracket\varphi\rrbracket_{\bar m,\bar n}$ of length at most $p(\max(\bar m,\bar n))$, for
    any $\bar m,\bar n$.
    \item If $\VNC^1\vdash\forall\bar X\forall \bar x\varphi(\bar x,\bar X)$, then there exist Frege
    proofs of all $\llbracket\varphi\rrbracket_{\bar m,\bar n}$ of length at most $p(\max(\bar m,\bar n))$, for
    any $\bar m,\bar n$.
  \end{iteMize}
  These proofs are effective in the sense that for any such $\varphi$ there exists a polynomial-time computable
  function $F_\varphi$ that maps any tuple $(\bar m,\bar n)$ to the above proofs of $\llbracket\varphi\rrbracket_{\bar m,\bar
  n}$.
\end{fact}

Facts~\ref{fact:VZ and bd Frege} and \ref{fact:VNC1 and NC1 Frege} are examples of general principles, the so
called {\em Reflection Principles}, which are defined as follows.

\begin{definition}
  [Reflection Principle] Let $P$ be a pps. Then the {\em Reflection Principle}
  for $P$, $\mathsf{Ref}_P$, is the $\forall\Delta^B_1$-formula (w.r.t. $\VZ$) $$\forall\Pi\forall
  X\forall Z ((\mathsf{Fla}(X)\land
  \mathsf{Prf}_P(\Pi,X))\rightarrow (Z\vDash X)),$$ where $\mathsf{Prf}_P$ is a $\Delta^B_1$-predicate formalizing
  $P$-proofs.
\end{definition}

Reflection Principles condense the strength of propositional proof systems. In what follows we will summarize
some such results for the proof systems and theories used here. A detailed exposition can be found in
\cite{CN10}, chapter X, or in \cite{Kra95}, chapter 9.3.

\begin{theorem}\label{Thm Simulation by Reflection}
  If $\VZ\vdash \mathsf{Ref}_{F}$ then bounded depth Frege
  p-simulates Frege.
\end{theorem}
We will only give a brief sketch of the proof here and leave out the technical details.
\begin{proof}[Sketch]
Let $\varphi$ be a formula and $\pi_{\varphi}$ a Frege proof of $\varphi$ which is witnessed by a Turing machine
$T_F$ (cf Def~\ref{def pps}). Since $\VZ$ proves $\mathsf{Ref}_{F}$, by Facts~\ref{fact:VZ and bd Frege} and
\ref{fact:polynomial simulation} we have polynomial size proofs of its translations $\llbracket
\mathsf{Ref}_F\rrbracket$ in bounded depth Frege. Bounded depth Frege itself, however, is strong enough to
verify that a proper encoding of the computation of $T_F$ on input $(\pi_\varphi,\varphi)$ is correct. Thus it
can verify that $\pi_{\varphi}$ is a Frege-proof and, using the translation of the Reflection Principle and the
Cut rule, conclude $\llbracket\mathsf{Taut}(\varphi)\rrbracket$. From this $\varphi$ follows, cf.~\cite{Kra95}
Lemma~9.3.7.
\end{proof}

Given a term $t$ and a variable $x$, we can also introduce the  $t$-bounded version of the Reflection Principle
for some given pps $P$, $\mathsf{Ref}_P(t(x))$ that claims soundness only for $t$-bounded proofs.

\begin{definition}
  [Bounded Reflection] Let $t$ be a $\mathcal L^2_{A}$-term, $x$
  a first-sort variable and $P$ a pps. Then the {\em Bounded
  Reflection Principle} $\mathsf{Ref}_P(t(x))$ is the formula
  $$\forall \Pi\leq t(x)\forall X\leq t(x)\forall Z\leq t(x)((\mathsf{Fla}(X)\land
  \mathsf{Prf}_P(\Pi,X))\rightarrow (Z\vDash X)).$$
\end{definition}

We can now generalize Theorem~\ref{Thm Simulation by Reflection} in the following way.

\begin{theorem}
  \label{Thm Simulation by Bounded Reflection}
  Let $t$ be a $\mathcal L^2_{A}$-term and $x$
  a number variable. If $t(x)<x$ for $x$ large enough
  and if $\VZ\vdash \forall x \mathsf{Ref}_{F}(t(x))$ then for every
  propositional formula $\varphi$ with a Frege proof of length
  $t(x)$ there is a bounded depth Frege proof of $\varphi$ of length
  $x^{O(1)}$. This proof can be efficiently constructed.
\end{theorem}

\begin{proof}
  The proof is the same as that of Theorem~\ref{Thm Simulation by
  Reflection}. Using the Bounded Reflection Principle we can encode
  Frege proofs of length $t(x)$ as bounded depth Frege proofs of length $x^{O(1)}$.
\end{proof}

As a corollary we get


\begin{corollary}
  If $\VZ\vdash \mathsf{Ref}_{F}(\abs{x}^k)$ for all $k\in\Nat$, then
  bounded depth Frege sub exponentially simulates Frege:
  For all $D>1,\delta>0$ exists $d\geq D$, such that the existence of a Frege proof of length $m$ of a depth $D$ formula
  implies the existence of a depth $d$ Frege proof of length at most $2^{m^{\delta}}$.
\end{corollary}

\section{Polylogarithmic Cuts of Models of $\VZ$ are Models of $\mathbf{VNC}^1$.}\label{Sec:Polylog Cuts}

We will first introduce the notion of a cut $\mathcal I$ of a given two-sorted arithmetic model $\mathcal M$.
This model theoretic approach provides a very good insight on what actually happens semantically with the small
elements of arithmetical models.

\begin{definition}
  [Cut] Let $T$ be a two-sorted arithmetic theory and
  $$N=\{N_1,N_2,+^{N},\cdot^{N},\leq^{N},0^N,1^N,
  \abs{\cdot}^{N},=^{N}_1,=^{N}_2,\in^{N}\}$$ a model of $T$.
  A {\em cut} $$M=\{M_1,M_2,+^{M},\cdot^{M},\leq^{M},0^M,1^M,
  \abs{\cdot}^{M},=^{M}_1,=^{M}_2,\in^{M}\}$$ of $N$ is any substructure
  such that
  \begin{iteMize}{$\bullet$}
    \item $M_1\subseteq N_1$, $M_2\subseteq N_2$,
    \item $0^M=0^N$, $1^M=1^N$,
    \item $M_1$ is closed under $+^{N},\cdot^{N}$ and downwards with respect to $\leq^{N}$,
    \item $M_2 =\{ X\in N_2\mid X\subseteq M_1\}$, and
    \item $\circ^{M}$ is the restriction of $\circ^{N}$ to $M_1$ and $M_2$ for all
          relation and function symbols $\circ\in\mathcal L_{A}^2$.
  \end{iteMize}
  We call this cut the {\em Polylogarithmic Cut} iff
  $$x\in M_1\Leftrightarrow \exists a\in N_1, k\in \Nat\ x\leq\abs{a}^k.$$
\end{definition}

To examine the strength of the theory of such cuts of models of $\VZ$, we will show that a formal connection
between efficient computability and $\Sigma^B_0$-definability holds. This stands in contrast to general bounded
subsets, where the connection is presumably only with respect to $\Sigma^B_1$-definability via the predicate
$\mathsf{ACC}$ (see \eqref{Def ACC} on page \pageref{Def ACC}). The intended theorem is a formalization of
Nepomnjascij's Theorem \cite{Nep70} (see also \cite{Kra95} pg.20). We will sketch the original proof before
starting the formalization.

\begin{theorem}[Nepomnjascij \cite{Nep70}]
  \label{Thm Nepomnjascij} Let $c\in\Nat$ and $0<\epsilon <1$ be
  constants. Then if the language $L\in \mathbf{TimeSpace}(n^c,n^{\epsilon})$,
  the relation $x\in L$ is definable by a $\Sigma^B_0$-formula over $\Nat$.
\end{theorem}

\proof
  We will prove the theorem by induction on $k$ for
  $L\in\mathbf{TimeSpace}(n^{k\cdot(1-\epsilon)},n^{\epsilon})$.

  Let $k=1$ and $L\in \mathbf{TimeSpace}(n^{k\cdot(1-\epsilon)},n^{\epsilon})$.
  For any $x\leq 2^n$ the whole computation can be coded by a number $y$ of size
  $2^{O(n)}$. The existence of such a computation gives the desired $\Sigma^B_0$-definability.

  For $k>1$ we write a sequence $y_0,y_1,\dots , y_{n^{1-\epsilon}}$
  of intermediate results coding the computation, where $y_0$ codes the starting configuration
  on input $x$, such that we can verify
  that $y_{i+1}$ is computable from $y_i$ in
  $\mathbf{TimeSpace}(n^{(k-1)\cdot(1-\epsilon)},n^{\epsilon})$.
  By assumption there exists a $\Sigma^B_0$-formula $\mathsf{reach}^{k-1}$ such that
  $\mathsf{reach}^{k-1}(y_i,y_{i+1})$ holds iff $y_{i+1}$ is computed
  from $y_i$. Additionally, the whole sequence has length $O(n)$
  and so we can write the sequence of intermediate results $y_i$ as a number $y$ of length $O(n)$.
  Now, the $\Sigma^B_0$-definition of $x\in L$ is simply
  \begin{equation*}
    \begin{split}
      \exists y\leq 2^{O(n)} \forall i&\leq n^{1-\epsilon} \mathsf{reach}^{k-1}(\langle y\rangle_i,\langle
      y\rangle_{i+1})\\&\land \langle y\rangle_0\mbox{ encodes the starting configuration of }A
      \mbox{ on input }x\\&\land \langle y\rangle_{n^{1-\epsilon}} \mbox{ is in an accepting
  state}.\rlap{\hbox to 161 pt{\hfill\qEd}}
    \end{split}
  \end{equation*}

\noindent We will now formalize this result in $\VZ$ as follows

\begin{theorem}\label{Thm Nepomnjascij_formalized} Let $N\vDash\VZ$.
  Let $m=\abs a$ for some $a\in N_1$ and let $c,k\in\Nat$ and $\epsilon<\frac{1}{k}$. If
  $L\in\mathsf{TimeSpace}(m^c,m^{\epsilon})$ (for strings of length $m$) is computed by Turing machine $A$, then there exists a $\Sigma^B_0$
  definition in $N$ of the $\Sigma^B_1$-predicate $\mathsf{ACC}_A$ on the interval $[0,m^k]$. I.e. any $Y\in L$,
  bounded by $m^k$ is $\Sigma^B_0$-definable in $N$ and therefore exists in the polylogarithmic cut of $N$.
\end{theorem}
  The following version of the proof stems from a discussion with Stephen Cook and Neil Thapen during the SAS
  programme in Cambridge. It is more explicit than the original one and clarifies the argument.

\begin{proof}
  We will inductively on $d$ define a $\Sigma^B_0$ relation $\mathsf{reach}^d_A(I,p_1,p_2,\mathsf{cell},\mathsf{comp})$
  that states that the $p_2$th cell of the work tape of $A$, starting on configuration $I$ and computing
  for $p_1\cdot m^{d\frac{1-k\epsilon}{k}}$ steps via the computation $\mathsf{comp}$ is $\mathsf{cell}$. We will bound the
  quantifiers in such a way that we can conclude that both variables can be of the number sort. As $d$ depends
  only on $A$ and $k$ we will be doing this induction outside of the theory to construct $d$ many formulas. We
  will then prove the above mentioned properties of $\mathsf{reach}^d_A$ by $\Sigma^B_0$ induction
  on $p_1$.

  Keep in mind that a cell is given as a pair $\langle bit, state\rangle$, where $bit$ is the actual
  value of the cell and $state$ is a number $>0$ coding the state the Turing machine is in iff the
  pointer is on that cell and $0$ otherwise. As before the transition function is denoted by $\delta$. We let
  $I$ be a string coding the input at the start of the computation. That is, $I$ is a sequence of length
  $len(I)\leq m^k$, such that $I[1,1]=1$, $I[j,1]=0$ for all $j>1$ and $I[j,0]$ is the $j$th input bit.
  We let $\langle\mathsf{comp}\rangle$ be a sequence encoding the computation of $A$, such that
  $\langle\mathsf{comp}\rangle_{\langle j,j',1\rangle}$ is the state, the machine is in after $j$ steps (0 denotes that the
  read/write head is not on cell $j'$, while a greater number gives the
  state and witnesses that the read/write head is on cell $j'$). $\langle\mathsf{comp}\rangle_{\langle j,j',0\rangle}$
  is the value of cell $j'$ after $j$ steps of the computation. Observe that this also implies that the
  computation can be encoded as a number, that is, it has to be very short. This is straight forward from the
  quantifier bounds.

  We can now define
\begin{equation*}
    \begin{split}
      &\mathsf{reach}^0_A(I,p_1,p_2,\mathsf{cell},\mathsf{comp})\\
&\equiv {\large (}\forall j'<\lceil len(I)^{\epsilon}\rceil,
      \langle\mathsf{comp}\rangle_{\langle 1,j',0\rangle} \approx  I[j',0]\land
      \langle\mathsf{comp}\rangle_{\langle 1,j',1\rangle} \approx
      I[j',1]{\large )}\land\\ 
&\phantom{{}\equiv{}}
      [\forall j<\lceil len(I)^{\frac{1-k\epsilon}{k}}\rceil,j'<\lceil
        len(I)^{\epsilon}\rceil,\alpha<\abs{A}\\
&\phantom{{}\equiv{}}\hspace{16pt}
      (\langle\mathsf{comp}\rangle_{\langle j,j',1\rangle} = 0\rightarrow (\langle\mathsf{comp}
      \rangle_{\langle j,j',0\rangle} = \langle\mathsf{comp}\rangle_{\langle
      j+1,j',0\rangle}))\land\\ 
&\phantom{{}\equiv{}}\hspace{16pt}
      (\langle\mathsf{comp}\rangle_{\langle j,j',1\rangle} =
      \alpha\rightarrow(\\
&\phantom{{}\equiv{}}\hspace{24pt}
      (\langle \delta(\alpha,\langle\mathsf{comp}\rangle_{\langle
        j,j',0\rangle})\rangle_3=0\rightarrow
      (\langle\mathsf{comp}\rangle_{\langle j+1,j',1\rangle}=\langle \delta(\alpha,\langle\mathsf{comp}\rangle_{\langle
      j,j',0\rangle})\rangle_1 \land\\
&\hspace{183 pt}
      \langle\mathsf{comp}\rangle_{\langle j+1,j',0\rangle}=\langle \delta(\alpha,\langle\mathsf{comp}\rangle_{\langle
      j,j',0\rangle})\rangle_2))\land\\
&\phantom{{}\equiv{}}\hspace{24pt}
      (\langle \delta(\alpha,\langle\mathsf{comp}\rangle_{\langle
        j,j',0\rangle})\rangle_3=1\rightarrow
      \langle\mathsf{comp}\rangle_{\langle j+1,j'\dotminus 1,1\rangle}=\langle \delta(\alpha,\langle\mathsf{comp}\rangle_{\langle
      j,j',0\rangle})\rangle_1 \land\\
&\hspace{179 pt}
      \langle\mathsf{comp}\rangle_{\langle j+1,j',0\rangle}=\langle \delta(\alpha,\langle\mathsf{comp}\rangle_{\langle
      j,j',0\rangle})\rangle_2))\land\\
&\phantom{{}\equiv{}}\hspace{24pt}
      (\langle \delta(\alpha,\langle\mathsf{comp}\rangle_{\langle
        j,j',0\rangle})\rangle_3=2\rightarrow
      \langle\mathsf{comp}\rangle_{\langle j+1,j'+1,1\rangle}=\langle \delta(\alpha,\langle\mathsf{comp}\rangle_{\langle
      j,j',0\rangle})\rangle_1 \land\\
&\hspace{179 pt}
      \langle\mathsf{comp}\rangle_{\langle j+1,j',0\rangle}=\langle \delta(\alpha,\langle\mathsf{comp}\rangle_{\langle
      j,j',0\rangle})\rangle_2))))\land\\
&\phantom{{}\equiv{}}\hspace{16pt}
      \forall\ell\neq\ell'<\lceil len(I)^{\epsilon}\rceil(
      \langle\mathsf{comp}\rangle_{\langle j,\ell,1\rangle}>0\rightarrow
      \langle\mathsf{comp}\rangle_{\langle
        j,\ell',1\rangle}=0)]\land\\
&\phantom{{}\equiv{}}
      (\langle \mathsf{cell}\rangle_1 = \langle\mathsf{comp}\rangle_{\langle p_1,p_2,0\rangle})\land
      (\langle \mathsf{cell}\rangle_2 = \langle\mathsf{comp}\rangle_{\langle p_1,p_2,1\rangle}).
    \end{split}
  \end{equation*}

  \noindent It is straightforward to prove by induction on the number of lines in $\mathsf{comp}$ that $\mathsf{comp}$
  is uniquely defined by
  $\mathsf{reach}^0_A$. We let $$\mathsf{Reach}^0_A(I,p_1,p_2,\mathsf{cell})=_{def}\exists\mathsf{comp}<q(\abs{I})\
  \mathsf{reach}^0_A(I,p_1,p_2,\mathsf{cell},\mathsf{comp}),$$
  where $q$ is some polynomial depending on the encoding. Here it is vital that $q$ can be defined such that
  $q(\abs{I})$ is a number in $N$. This is possible due to the quantifier bounds we used when defining
  $\mathsf{reach}^0_A$. Thus, $\mathsf{Reach}^0_A$ is defined by a $\Sigma^B_0$ formula.

  Informally $\mathsf{Reach}^0_A$ formalizes that there is a computation
  $$
  \begin{matrix}
    \langle\mathsf{comp}\rangle_{\langle 1,1,\cdot\rangle} & \langle\mathsf{comp}\rangle_{\langle 1,2,\cdot\rangle} & \cdots &\langle\mathsf{comp}\rangle_{\langle 1,(m^k)^\epsilon,\cdot\rangle}\\
    \langle\mathsf{comp}\rangle_{\langle 2,1,\cdot\rangle} & \langle\mathsf{comp}\rangle_{\langle 2,2,\cdot\rangle} & \cdots &\langle\mathsf{comp}\rangle_{\langle 2,(m^k)^\epsilon,\cdot\rangle}\\
     \vdots    & \ddots & & \vdots\\
     \vdots    & & \ddots & \vdots\\
    \langle\mathsf{comp}\rangle_{\langle (m^k)^{\frac{1-k\epsilon}{k}},1,\cdot\rangle} & \langle\mathsf{comp}\rangle_{\langle (m^k)^{\frac{1-k\epsilon}{k}},2,\cdot\rangle} & \cdots &\langle\mathsf{comp}\rangle_{\langle (m^k)^{\frac{1-k\epsilon}{k}},(m^k)^\epsilon,\cdot\rangle}
  \end{matrix}
  $$
  that is correct in the sense that we can verify that we get from line to line via the transition function of $A$ and gives the appropriate values of the cell at $(p_1,p_2)$. Observe that the size of the whole computation as presented above is linear in $m$, i.e. that it can be coded as a number in $N$.

  We will now proceed by inductively defining $\mathsf{reach}^d_A$ and $\mathsf{Reach}^d_A$. Assume that
  $\mathsf{reach}^{d\dotminus 1}_A$ has already been defined by a $\Sigma^B_0$ formula over $N$. We then let
  \begin{equation*}
    \begin{split}
      \mathsf{reach}^d_A(I,p_1,&\,p_2,\mathsf{cell},\mathsf{comp})\\
 &\equiv  (\forall j'<\lceil len(I)^{\epsilon}\rceil,
      \langle\mathsf{comp}\rangle_{\langle 1,j',0\rangle} \approx  I[j',0]\land
      \langle\mathsf{comp}\rangle_{\langle 1,j',1\rangle} \approx
      I[j',1])\land\\
 &\phantom{{}\equiv{}}
      (\forall j<\lceil len(I)^{\frac{1-k\epsilon}{k}}\rceil\exists\mathsf{comp'}<q(len(I))\forall j'<\lceil
      len(I)^{\epsilon}\rceil\exists \mathsf{cell}'<\abs{A}\forall
      j''<\abs{A}\\
&\phantom{{}\equiv{}}
      (\langle\mathsf{comp}\rangle_{\langle
        j+1,j',j''\rangle}\leftrightarrow \langle
      \mathsf{cell}'\rangle_{j''})\land\\
&\phantom{{}\equiv{}}
      \mathsf{reach}^{d\dotminus 1}_A(\langle \mathsf{comp}\rangle_{\langle
      j,\cdot,\cdot\rangle},m^{\frac{1-k\epsilon}{k}},
      j',\mathsf{cell}',\mathsf{comp'}))\land\\
&\phantom{{}\equiv{}}
      (\langle \mathsf{cell}\rangle_1\leftrightarrow\langle\mathsf{comp}\rangle_{\langle p_1,p_2,0\rangle})\land
      (\langle \mathsf{cell}\rangle_2=\langle\mathsf{comp}\rangle_{\langle p_1,p_2,1\rangle}).
    \end{split}
  \end{equation*}
  Again, we can prove uniqueness of the computation by induction on the number of its lines and let
  $$\mathsf{Reach}^d_A(i,p_1,p_2,\mathsf{cell})=_{def}\exists\mathsf{comp}<q(\abs{I})\
  \mathsf{reach}^d_A(i,p_1,p_2,\mathsf{cell},\mathsf{comp}).$$ That this is a $\Sigma^B_0$ definition follows by
  induction and the same argument as for $\mathsf{Reach}^0_A$. Here, the predicate $\mathsf{Reach}^{d-1}_A(i,p_1,p_2,\mathsf{cell})$ takes the role of the transition function in witnessing that each line follows from the preceding one. The total size again is linear in $m$.

  We now can give a
  $\Sigma^B_0$ definition of the predicate $W[\langle i,j,\cdot\rangle]$ coding the computation as in $\mathsf{ACC}_A$ on input $X$ of length $m^k$. We let $W[\langle i,j,\cdot\rangle]=\mathsf{cell}\equiv$
  \begin{equation*}
    \begin{split}
&
      \exists r_0,\dots,r_d<\abs{X}^{\frac{1-k\epsilon}{k}},con_1,\dots,con_d<p(\abs{X}^{\epsilon})
      \forall z_1,\dots,z_d<\abs{X}^{\epsilon}\exists
      \mathsf{cell}_1,\dots,\mathsf{cell}_d<\abs{A}\\
&\hspace{8 pt}\phantom{{}\equiv{}}
      (i=\smash{\sum_{\ell=0}^d} r_{\ell}\cdot
      \abs{X}^{\ell\frac{1-k\epsilon}{k}}
      \land \mathsf{Reach}^d_A(\tilde
      X,r_d,z_d,\mathsf{cell}_d)\land\langle
      con_d\rangle_{z_d}=\mathsf{cell}_d\\
&\hspace{118 pt}
      \land
      \mathsf{Reach}^{d-1}_A(con_d,r_{d-1},z_{d-1},\mathsf{cell}_d)\land\langle
      con_{d-1}\rangle_{z_{d-1}}=\mathsf{cell}_{d-1}\\
&\hspace{118 pt}
      \ \vdots\\
&\hspace{118 pt}
      \land \mathsf{Reach}^{0}_A(con_1,r_{0},j,\mathsf{cell})),
    \end{split}
  \end{equation*}
  where $p$ is a polynomial depending on the encoding and $\tilde X$ is the starting configuration of $A$ on input
  $X$.

  Informally the above formula says that we compute the configurations of $A$ by using the predicates
  $\mathsf{Reach}^d_A$ through $\mathsf{Reach}^0_A$. That is, after the application of $\mathsf{Reach}^d_A$
  (i.e. after making the biggest steps) we have reached
  configuration $con_d$, which we plug into $\mathsf{Reach}^{d-1}_A$ to get configuration $con_{d-1}$ and so on.
  It remains to show that this definition of $W[\langle
  i,j,\cdot\rangle]$ coincides with the real one, i.e. that $W[\langle i+1,\cdot,\cdot\rangle]$ follows from an application of the transition
  function of $A$ from $W[\langle i,\cdot,\cdot\rangle]$.

  We will prove this inductively, depending on $i$. Again let $r_{\ell}$ be such that $i=\sum_\ell r_{\ell}\cdot
  \abs{X}^{\ell\frac{1-k\epsilon}{k}}$.
  If $i<\abs{X}^{\frac{1-k\epsilon}{k}}$ the assumption follows straightforwardly from the definition of
  $\mathsf{Reach}^0_A$. Now for bigger $i$. If the $r_0$, given as above, is bigger then $0$ then
  again the assumption follows from the definition of $\mathsf{reach}^0_A$. Now let $\ell>0$ be the first index
  with $r_{\ell}>0$. We the have to argue that $\mathsf{reach}^{\ell'-1}_A$ has the desired property.
  This, however, follows straightforward if we can verify this assertion for $\mathsf{reach}^{\ell'-1}_A$.
  Observe that $d$ is a constant depending only on $A$ and $k$, so we need to make this argument only a
  constant number of steps to reach $\mathsf{reach}^0_A$, where we know that the assertion holds.

  Since we can code the whole computation as a $\Sigma^B_0$-formula (in $N$), we can easily deduce a
  $\Sigma^B_0$-definition of the related set by simply stating that the computation accepts (i.e. that in the
  last line of the computation the state is accepting). This concludes the proof.
\end{proof}

We can now prove our main result.
\begin{theorem}\label{thm:conclusion Nepo}
  Let $N\vDash \VZ$ and $M\subseteq N$ be the polylogarithmic cut. Then
  $M\vDash \VNC^1$.
\end{theorem}

\begin{proof}
  We have to prove that for all strings $G_{\varphi}\in M_2$, representing a formula $\varphi$ as a tree and
  assignments $I\in M_2$ to its variables (i.e. leafs in the tree representation) a string $Y$ exists in $M_2$ that
  contains all values of $\varphi$'s subformulas as in the definition of $MFV$ in Section~\ref{Sec:Ext of VZ}
  and satisfies the inductive conditions of $MFV$.


  However, by $\Sigma^B_0$-comprehension and the formalized
  Nepomnjascij's Theorem it suffices to describe an algorithm that
  computes, for given $G_{\varphi}$ and $I$, whether $i$ belongs to $Y$  in
  $\mathsf{TimeSpace}(\abs{G_{\varphi}}^k,\abs{G_{\varphi}}^{\epsilon})$
  for some $k\in\Nat$ and $\epsilon<1$.

The following is a recursive algorithm computing the value of $Y[i]$, given $G:=G_{\varphi},I$ and $i$.

\verb"NodeValue(G,I,i)"

{\small
\begin{itemize}
  \item \verb"boolean left; boolean right;"
  \item \verb"If i>2"$\cdot$\verb"|G|"
  \begin{itemize}
    \item \verb"Output (0); End;"
  \end{itemize}
  \item \verb"Else If i>|G|"
  \begin{itemize}
    \item \verb"Output (I[i-|G|]); End;"
  \end{itemize}
  \item \verb"Else If G[i]=1"
  \begin{itemize}
    \item \verb"left := NodeVal(G,I,2i);"
    \item \verb"right := NodeVal(G,I,2i+1);"
    \item \verb"Output (left AND right); End;"
  \end{itemize}
  \item \verb"Else If G[i]=0"
  \begin{itemize}
    \item \verb"left := NodeVal(G,I,2i);"
    \item \verb"right := NodeVal(G,I,2i+1);"
    \item \verb"Output (left OR right); End;"
  \end{itemize}
    \item \verb"Else"
  \begin{itemize}
    \item \verb"Output (0); End;"
  \end{itemize}
\end{itemize}
} 
\noindent Observe that the algorithm at any given point only stores a
constant amount of data per level of the tree $G$ and therefore uses
only $O(\log(\abs{G}))$ space. The number steps the algorithm makes is
clearly polynomial in the size of $G$. Therefore by Theorem~\ref{Thm
  Nepomnjascij_formalized}, for every monotone formula $\varphi$,
representable as a tree in $M$, we get a $\Sigma^B_0$ formula
$\mathsf{eval}_{\varphi}$, such that
$\mathsf{eval}_{\varphi}(i,I)\equiv Y[i]$. Observe that
$\mathsf{eval}_{\varphi}$ depends on the size of $\varphi$ and on its
logical depth, as the first is essentially the size of the input for
the machine, that $eval_{\varphi}$ codes, while the latter determines
the longest iterations in the recursive algorithm. Applying the
Comprehension Schema in $\VZ$, i.e. in $N$, this verifies the
existence of a $Y$ as in $MFV$ for all formulas represented by trees
in $M$.Therefore $MFV$ holds in $M$ and so $M\vDash\VNC^1$.
\end{proof}

\section{Implications for Proof Complexity}\label{Sec:Impl Proof Cmpl}

We now wish to apply the above results to propositional proof systems. More precisely we wish to show that
theories of small cuts of a model of a given theory $\mathcal T$ correspond to stronger proof systems than
$\mathcal T$ does. An elegant way of showing such a statement is via the {\em Reflection Principles} of the
given proof systems, i.e. the statement that the proof system is correct, as explained in Section~\ref{Sec:Rel
BA pps}. With their help we can conclude the following recent result of Filmus, Pitassi and Santhanam
\cite{FPS11}.

\begin{theorem}[\cite{FPS11}]\label{thm FilPitSat}
  Every Frege system is sub exponentially simulated
  by $\mathsf{AC}^0$-Frege systems.
\end{theorem}

\begin{proof}
  By Theorem~\ref{Thm Simulation by Bounded Reflection} we have to prove the polylogarithmically bounded Reflection
  Principle for Frege in $\VZ$. This, by Theorem~\ref{thm:conclusion Nepo} however,
  corresponds to proving the Reflection Principle for Frege in $\VNC^1$, which holds by
  Fact~\ref{fact:VNC1 and NC1 Frege}. It also follows from Theorem~\ref{Thm Simulation by Bounded Reflection} that this proof can be efficiently computed from the Frege proof.
\end{proof}

Another, related, application is in the separation of propositional proof systems. In \cite{MT11} we proved the
following.
\begin{proposition}\label{Prop Iddo me}
  For almost every random 3CNF $A$ with $n$ variables and $m=c\cdot n^{1,4}$ clauses, where $c$ is a large
  constant, $\neg A$ has polynomially bounded $\mathsf{TC}^0$-Frege proofs.
\end{proposition}
On the other hand it is well known (see for example \cite{CS88}) that such formulas have no subexponential
refutation in Resolution. Thus, this yields an average case separation between Resolution and
$\mathsf{TC}^0$-Frege. We can now extend this result to an average case separation between Resolution and
$\mathsf{AC}^0$-Frege as follows.
\begin{theorem}
  For almost every random 3CNF $A$ with $n$ variables and $m=c\cdot n^{1,4}$ clauses, where $c$ is a large
  constant, $\neg A$ has subexponentially bounded $\mathsf{AC}^0$-Frege proofs.
\end{theorem}
\begin{proof}
  By Theorem~\ref{thm:conclusion Nepo} the polylogarithmic Cut of any $\VZ$-model is a model of $\VNC^1$,
  therefore also of $\mathbf{VTC}^0$. This yields, as in our proof of Theorem~\ref{thm FilPitSat}, that
  $\mathsf{AC}^0$-Frege subexponentially simulates $\mathsf{TC}^0$-Frege. The result now follows from
  Proposition~\ref{Prop Iddo me}.
\end{proof}

\section{Conclusion and Discussion}

As we have seen cuts of models of weak arithmetics constitute an appropriate way for reasoning about
super-polynomial simulations between proof systems. An advantage in comparison to syntactic arguments is the
possible applicability of results in Model Theory and a more uniform treatment. This can readily be observed as
with our argument, e.g. the work of Paris and Wilkie \cite{PW85}\cite{PW87} immediately imply the simulation
results from Bonet et al. \cite{BDGMP04}.

This leads to interesting possibilities for further research, especially towards the weak automatizability of
weak propositional proof systems such as Resolution. The underlying theory, which was $\VZ$ in our argument,
must be significantly weakened, however. If we could take $T^2_1(\alpha)$ as our base theory, we could reason
about whether $Res(\log)$ has the feasible interpolation property in the same way as Kraj\' i\v cek and Pudl\'
ak \cite{KP98}, Bonet, Pitassi and Raz \cite{BPR00} or Bonet, Domingo, Gavald\`a, Maciel, and Pitassi
\cite{BDGMP04}. Now, if $Res(\log)$ does not have quasi-polynomial feasible interpolation we know by a result
from Atserias and Bonet \cite{AB04} that Resolution is not weakly automatizable, so we would be finished.
Whether we can actually do it depends on the strength of the theory the polylogarithmic cut of $T^2_1(\alpha)$
models and if we can formalize some sort of iterated multiplication (such as in \cite{ABH02}) in that theory.
Also, the security of Diffie-Hellman seems to be a more appropriate assumption than that of RSA, as the
computational power needed to verify the correctness of Diffie-Hellman seems to be lower.

\section{Acknowledgements}

I want to thank Steve Cook, Jan Kraj\' i\v cek and Neil Thapen for helpful suggestions and discussion, Emil Je\v
r\' abek for his comments and for answering my questions and the participants of the MALOA Special Semester in
Proof Complexity in Prague 2011 for enduring a sloppy and sometimes faulty exposition of this proof and still
coming up with helpful comments. I also want to thank the anonymous referees for pointing out various mistakes
and for giving interesting suggestions.
A similar construction can be extracted from \cite{Zam97} and leads to similar results, if perceived in the way we did it here. I want to thank Leszek Kolodziejczyk for pointing this out.


\end{document}